\documentclass[journal]{IEEEtran}

\usepackage{latexsym}
\usepackage{graphicx}
\usepackage{array}
\usepackage{amsmath}
\usepackage{amsfonts}
\usepackage{amssymb}
\usepackage{amsthm}

\newcommand{\bi}{\begin{itemize}}
\newcommand{\ei}{\end{itemize}}
\newcommand{\ben}{\begin{enumerate}}
\newcommand{\een}{\end{enumerate}}

\newtheorem{thm}{Theorem}
\newtheorem{lemma}{Lemma}
\newtheorem{prop}{Proposition}
\newtheorem{defn}{Definition}

\newcommand{\by} {\boldsymbol{y}}

\newcommand{\bh} {\boldsymbol{h}}

\newcommand{\bw} {\boldsymbol{w}}

\newcommand{\gl}{\lambda}

\newcommand{\bxi} {\boldsymbol{\xi}}

\def\bal#1\eal{\begin{align}#1\end{align}}
\newcommand{\bp} {\begin{proof}}
\newcommand{\ep} {\end{proof}}

\newcommand{{\bRF}} {\right\}}

\begin{document}

\title{Favorable Propagation for Massive MIMO with Circular and Cylindrical Antenna Arrays}

\author{Elham Anarakifirooz, Sergey Loyka

\vspace*{-1\baselineskip}

\thanks{E. Anarakifirooz and S. Loyka are with the School of Electrical Engineering and Computer Science, University of Ottawa, Ontario, Canada, e-mail: sergey.loyka@ieee.org}
\thanks{This is an extended version of the paper accepted by IEEE Wireless Comm. Letters, Nov. 2021.}
}

\maketitle

\vspace*{-1\baselineskip}
\begin{abstract}
Massive MIMO systems with uniform circular and cylindrical antenna arrays are studied. Favorable propagation property is rigorously shown to hold asymptotically in LOS environment for a fixed antenna spacing and under some mild conditions. The analysis is based on a novel method using a Bessel function expansion, a new bounding technique and a simplified representation of inter-user interference for array geometries obeying Kronecker structure.
\end{abstract}


\vspace*{-0.5\baselineskip}
\section{Introduction}

Massive MIMO systems have recently gained significant attention and are considered to be a key technology for 5G and beyond, due to their promise of high spectral and energy efficiencies and simplified processing in multi-user environments. These promises rely to a significant degree on mutual orthogonality of channels to different users, which improves as the number of base station (BS) antennas grows. This is known as "favorable propagation" (FP) \cite{Marzetta-16}\cite{Ngo-13}.

FP was studied both theoretically \cite{Marzetta-16}-\cite{Masouros-15} and experimentally \cite{Hoydis-12}-\cite{Gauger-15}. While theoretical studies rely on some simplifying assumptions, such as i.i.d. Rayleigh fading or free-space propagation, measurement-based studies do take into account many practically-important factors, such as real-world wave propagation, antenna array geometry/design, etc. Despite the fact that FP never holds exactly in practice (i.e. channels of different users are never exactly orthogonal to each other), measurement-based studies show that using a large number of antennas allows one to recover a significant portion of theoretically-established benefits of massive MIMO systems in various real-world environments \cite{Hoydis-12}-\cite{Gauger-15}, even with reasonably-large number of antennas, since user orthogonality improves as the number of antennas increases.

It was shown via analysis, based on the law of large numbers, that FP is approached asymptotically (as the number of antennas grows) in i.i.d. fading channels \cite{Telatar-95}\cite{Marzetta-16}\cite{Ngo-13}. The i.i.d. assumption, however, disregards the impact of antenna array geometry and its constraints (such as its size). Yet, antenna array geometry is known to have a significant impact on MIMO system performance, especially when scattering is not rich enough. This impact was studied in \cite{Ngo-14}, where it was shown that FP holds asymptotically for uniform linear arrays (ULA) of fixed antenna spacing, for both fixed and uniformly-random users' angles-of-arrival (AoA) under line-of-sight (LOS) propagation. When the ULA size is fixed (rather than antenna spacing), FP does not hold anymore under LOS propagation and uniformly-random users' AoAs \cite{Masouros-15}.

It is concluded in \cite{Chen-13} that FP holds asymptotically for ULA and uniform planar arrays (UPA) but not for uniform circular arrays (UCA), all under LOS propagation. However, this comparison is not fair since ULA and UPA are analyzed under fixed antenna spacing (so their size grows unbounded with the number of antennas) while UCA is analyzed under fixed array size. Whether FP holds or not for UCA of fixed antenna spacing remains unknown. The present paper settles this open problem. We also consider a uniform cylindrical array (UCLA), which is a popular array geometry for practical implementations \cite{Hoydis-12}-\cite{Gunnarsson-20} but for which FP has not been analyzed yet. It should be emphasized that the analysis of UCA/UCLA is more challenging than that of ULA/UPA since, unlike the latter case, no closed-form expressions for array patterns are available in the former case and thus new tools are need to analyze their FP properties.

Since UCA and UCLA are attractive solutions for 5G/6G cellular applications, we analyze them here under fixed antenna spacing and LOS propagation. The latter is motivated by mmWave/THz systems (another key technology for 5/6G and beyond), where a large number of antennas can be located in a limited space but where LOS environment is essential to maintain a proper SNR (since a significant LOS blockage results in a large SNR loss resulting in link outage) \cite{Larsson-18}-\cite{Rappaport-19}. Furthermore, LOS environment is considered to be "particularly difficult" for users' orthogonality and massive MIMO performance \cite{Gao-15}. Therefore, it is important to know whether mmWave/THz systems can still exploit fully the advantages of massive MIMO in this environment.

It should be pointed  out that the method of \cite{Chen-13}, which was used for the UCA of a fixed size, is not applicable when the antenna spacing is fixed, since the respective sum does not converge to an integral under fixed spacing. Furthermore, the law of large numbers cannot be used either in deterministic setting considered here. Therefore, a new approach is needed. To this end, we propose such an approach based on a Bessel series representation. By carefully bounding each term of this series, we rigorously demonstrate that FP does hold asymptotically for UCA and UCLA of fixed antenna spacing under LOS propagation, see Theorem 1 and Propositions 1-3. This is in line with the ULA and UPA geometries, even though the analytical tools used here for the UCA/UCLA geometries are completely different from those of the ULA/UPA. Based on the asymptotic analysis, a condition on the number of antennas to approach the FP closely is given. In the process, we establish a new bound on Bessel functions and a new inter-user interference (IUI) representation for array geometries obeying Kronecker structure (such as UCLA or UPA), which significantly simplifies the analysis. The reported results establish point-wise convergence to orthogonality (i.e. for any given user) and not only statistically (i.e. "on average"), as in e.g. \cite{Marzetta-16}\cite{Ngo-13}\cite{Ngo-14}\cite{Masouros-15}\cite{Telatar-95}\cite{Loyka-19}. Statistical convergence follows from point-wise one under some mild assumptions on users' AoA distribution.

\section{Channel Model}
\label{sec.Channel Model}

Let us consider the uplink of a Gaussian MIMO channel, where $M$ independent single-antenna users transmit data simultaneously to a base station (BS) equipped with $N$-element antenna array:
\vspace*{-.5\baselineskip}
\bal
\label{eq.ch.model}
\by = \bh_1 x_1 + \sum_{i=2}^{M} \bh_i x_i + \bxi
\eal
where $\bh_i,\ x_i$ are the channel vector and information-bearing signal of $i$-th user, $i=1,...,M$; $\by,\ \bxi$ are the received BS signal and noise vectors, respectively; $|\bh|$, $\bh'$ and $\bh^+$ denote Euclidean norm (length), transposition and Hermitian conjugation, respectively, of vector $\bh$. The noise is Gaussian circularly-symmetric, of zero mean and variance $\sigma_0^2$ per Rx antenna. The channel vectors are assumed to be fixed, which corresponds to LOS-dominated propagation, as in mmWave systems \cite{Rangan-14}\cite{Larsson-18}.

To simplify system design, the BS uses linear processing: while decoding user $1$ signal (the main user), it treats the other users' signal $\sum_{i=2}^{M} \bh_i x_i$ as interference and applies the matched filter beamforming $\bw= \bh_1/|\bh_1|$ tuned to the main user so that its SINR can be expressed as follows:
\bal
\label{eq.SINR}
\mbox{SINR} &=\frac{|\bh_1|^2\sigma_{x_1}^2}{|\bh_1|^{-2}\sum_{i=2}^{M}|\bh_1^+\bh_i|^2\sigma_{x_i}^2 + \sigma_0^2}\\
&= \frac{\gamma_1}{\sum_{i=2}^{M}|\alpha_{iN}|^2\gamma_i + 1},\quad \alpha_{iN}=\bh_1^+\bh_i/N
\eal
where $\sigma_{xi}^2$ and $\gamma_i = |\bh_i|^2\sigma_{xi}^2/\sigma_0^2$ are the Tx signal power and the Rx SNR of user $i$; the channel is normalized so that $|\bh_i|^2=N$ (the propagation path loss is absorbed into the Rx SNR $\gamma_i$). Note that the matched-filter beamforming via $\bw$ is also known as (single-user) maximum ratio combining and $\gamma_1$ is its output SNR. $\sum_{i=2}^{M}|\alpha_{iN}|^2\gamma_i$ represent inter-user interference (IUI), which disappears if
\bal
|\alpha_{iN}|=0, \ i=2,...,M
\eal
where $|\alpha_{iN}|^2$ characterises IUI power "leakage" from user $i$ to the main user. Note that, in general,
\bal
\mbox{SINR} \le \gamma_1
\eal
i.e. the SINR cannot exceed the single-user SNR $\gamma_1$; the upper bound is attained and the two are equal
\bal
\mbox{SINR} = \gamma_1 \quad \mbox{if} \quad |\alpha_{N}|^2 =\sum_{i=2}^M |\alpha_{iN}|^2=0
\eal
This is the most favorable condition in terms of IUI, and one of the key motivations for massive MIMO systems. While $\alpha_N$ is never exactly zero in practice, it can approach zero as $N$ increases, which is known as (asymptotically) favorable propagation \cite{Marzetta-16}. This is made precise below.

\begin{defn}
Favorable propagation is said to hold if
\bal
\label{eq.FP}
\lim_{N\to\infty}|\alpha_{N}|=0
\eal
so that the inter-user interference asymptotically vanishes under bounded Rx SNR (i.e. can be made as low as desired provided the number of antennas is large enough).
\end{defn}

\section{FP for Uniform Circular Arrays}
\label{sec.FP.UCA}

In this section, we first analyze FP for a finite number of users with fixed and distinct AoA's and show that favorable propagation does hold for uniform circular arrays with any fixed element spacing $d>0$ under LOS propagation. Then, we briefly discuss the case of unbounded number of users and identify some scenarios where FP does not hold.

Note that the method of \cite{Chen-13}, which is based on the integral sum representation, cannot be used under fixed $d$ (since, in this case, the summation does not converge to an integral). Hence, new tools are needed. To this end, we will use the generating function of Bessel coefficients and establish novel upper bounds, from which the FP will follow without relying on the integral sum representation.

\subsection{Planar 2-D case}

Since the analysis is rather involved, we begin with the planar (2-D) case, when users are located in the plane of the UCA, see Fig. \ref{fig.UCA-2D} for an illustration, and then extend this to the full 3-D case.
\begin{figure}[t]
\centerline{\includegraphics[width=2in]{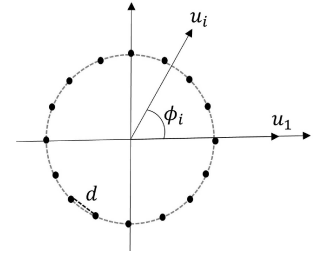}}
\vspace*{-.5\baselineskip}
\caption{A planar (2-D) case: users are located in the UCA plane.}
\label{fig.UCA-2D}
\end{figure}
In this 2-D case under LOS propagation in the far field, the normalized channel vectors can be expressed as\footnote{We omit here the common phase shift term $e^{j\varphi_i}$ since it does not affect FP; the LOS path loss is absorbed into the Rx signal power/SNR.}
\bal\notag
\bh_i &= [e^{j \psi_{0i}},..., e^{j \psi_{(N-1)i}}]',\ i=1...M\\
 \psi_{ni} &= 2\pi R_N\cos(\phi_i- \varphi_n),\quad   n=0...N-1
\eal
where $\varphi_n = 2\pi n/N$, $R_N = d|2\sin(\varphi_1/2)|^{-1}$ is the circle radius, $d$ is the element spacing (normalized to the wavelength), $\phi_i$ is the angle-of-arrival (AoA) of user $i$. The FP property for this array geometry is established below.

\begin{thm}
\label{thm.2D}
Favorable propagation holds for the UCA under LOS propagation, (any) fixed element spacing $d>0$ and any fixed number of users $M < \infty$ with  distinct AoA's, i.e.
\bal
\label{eq.FP-2D}
\phi_1 \neq \phi_i,\quad i=2...M, \quad \Leftrightarrow\quad \lim_{N\to\infty}|\alpha_{N}|=0
\eal
in the planar 2-D case.
\end{thm}
\begin{proof}
See Appendix.
\end{proof}

It should be pointed out that Theorem 1 establishes that the FP holds point-wise, i.e. for any given user (which can be taken to be user 1, without loss of generality) under distinct AoAs, and not only statistically, as in the case of fading channels and randomly-located users. Using the point-wise convergence above, one can establish the FP for randomly-located users as well, provided that the AoA distribution function satisfies certain mild conditions.

It follows from the proof of Theorem 1 that $|\alpha_{iN}| \ll 1\ \forall i\ge 2$ , i.e. inter-user interference leakage is low, if
\bal
N d \gg \left(\min_{i\ge 2} |\sin(\phi_i/2)| \right)^{-1}
\eal
where we set $\phi_1=0$. This determines the minimum number of antennas for low IUI leakage. If the users' AoAs are distributed uniformly, then $\phi_2 = 2\pi/M$ is the minimizer and, for large $M$, this becomes
\bal
N \pi d \gg M
\eal
provided fixed $d$ is not too large, $e\pi d < M$.

If the number of users $M$ is allowed to grow unbounded with $N$, then the FP does not hold anymore in general (but may hold in some special cases). To see this, consider $M=N$ users with uniformly-distributed AoAs so that $\phi_i = 2\pi(i-1)/N,\ i=1...N$, and use the proof of Theorem 1 to show that, in this case,
\bal
\label{eq.alpN.e1}
\lim_{N\to \infty} \alpha_N \ge \lim_{N\to \infty} |\alpha_{2N}| = |J_0(2\pi d)|
\eal
and $J_0(2\pi d) \neq 0$ unless $2\pi d$ is a null of $J_0(x)$, where $J_0$ is Bessel function of the 1st kind and order 0. For example, $J_0(2\pi d)=J_0(\pi)\approx -0.3$ if $d=1/2$, so there is no FP here. \eqref{eq.alpN.e1} also holds for arbitrary AoAs distribution provided that the nearest user (to user 1) is at $2\pi/N$. FP may also not hold even for a finite and fixed $M$ if AoAs are not fixed but some of them are allowed to approach that of user 1 arbitrarily closely as $N$ grows. For example, let $M=2$ and $\phi_2 = 2\pi/N$ so that
\bal
\lim_{N\to \infty} \alpha_N =|J_0(2\pi d)| \neq 0
\eal
i.e. the FP does not hold. Detailed study of these settings is beyond the scope of this Letter.


\subsection{Extension to 3-D case}

Next, we consider the full 3-D case, see Fig. \ref{fig.UCA-3D}.

\begin{figure}[t]
\centerline{\includegraphics[width=2.5in]{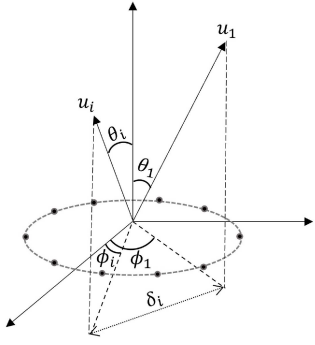}}
\caption{The full 3-D case geometry: users are not confined to the UCA plane.}
\label{fig.UCA-3D}
\end{figure}

\begin{prop}
\label{prop.3D}
Favorable propagation holds for the UCA in the full 3-D case under LOS propagation, for (any) fixed element spacing $d>0$, any fixed $M< \infty$, and any distinct and fixed AoA's if and only if neither of the following holds for any $i\ge 2$:

1. $\theta_{1,i} = 0$ or $\pi$

2. $\phi_1 = \phi_i$ and $\theta_1=\theta_i$ or $\theta_1=\pi-\theta_i$

\end{prop}
\begin{proof}
Follows along the lines of that of Theorem \ref{thm.2D}, with proper modifications to accommodate 3-D setting. In particular, the channel vector $\bh_i$ of user $i$ can be expressed as
\bal\notag
\bh_i &= [e^{j \psi_{0i}},..., e^{j \psi_{(N-1)i}}]',\ i=1...M\\
 \psi_{ni} &= 2\pi R_N\sin(\theta_i)\cos(\phi_i- \varphi_n),\ n=0...N-1
\eal
where $\varphi_n = 2\pi n/N$; $\alpha_{iN}$ can be expressed as
\bal\notag
\alpha_{iN} &=\frac{1}{N}\sum_{n=0}^{N-1} e^{j 2\pi R_N(\delta_{1i} \cos \varphi_n +\delta_{2i} \sin \varphi_n)}\\
&=\frac{1}{N}\sum_{n=0}^{N-1} e^{j 2\pi R_N\delta_i\sin(\varphi_n +\beta_i)}\\ \notag
&= J_0(z_{Ni}) + \sum_{k=1}^{\infty} (e^{jkN\beta_i} + (-1)^{kN} e^{-jkN\beta_i}) J_{kN}(z_{Ni})
\eal
where $z_{Ni}= 2\pi R_N \delta_i$, $\delta_i= \sqrt{\delta_{1i}^2 + \delta_{2i}^2}$, $\beta_i= \arg(\delta_{2i}+j \delta_{1i})$,
\bal\notag
\delta_{1i}=\sin \theta_{i} \cos \phi_{i}-\sin \theta_{1} \cos \phi_{1} \\
\delta_{2i}=\sin \theta_{i} \sin \phi_{i}-\sin \theta_{1} \sin \phi_{1}
\eal
Now, $|\alpha_{iN}|$ can be upper bounded as follows:
\bal\notag
|\alpha_{iN}| &\le 2 \sum_{k=0}^{k_1-1} |J_{k N}(z_{Ni})| +2 \sum_{k=k_1}^{\infty} |J_{k N}(z_{Ni})|\\
&= A_{1N,i} + A_{2N,i}
\eal
Using \eqref{eq.A1N.UB} and \eqref{eq.A2N.UB}, the desired result follows provided $\delta_i \neq 0$, which gives the conditions of Proposition \ref{prop.3D}.
\end{proof}


\section{FP for Uniform Cylindrical Arrays}
\label{sec.FP.UCLA}

In this section, we consider a $(N_c\times N)$-element uniform cylindrical array (UCLA) consisting of $N_c$ UCAs located on top of each other, each with $N$ elements, as shown in Fig. \ref{fig.UCLA}.
\begin{figure}[t]
\centerline{\includegraphics[width=2.7in]{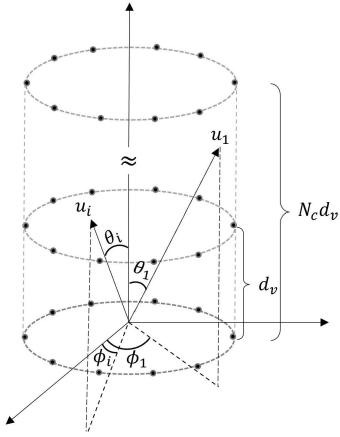}}
\caption{Uniform cylindrical array (UCLA) comprised of $N_c$ UCAs located on top of each other.}
\label{fig.UCLA}
\end{figure}
The aggregate channel vector $\bh_i$ of user $i$ can be expressed as follows:
\bal
\label{eq.hi.Kron}
\bh_i &= [\bh_{i1}',...,\bh_{iN_c}']' = \bh_{ih} \otimes \bh_{iv}
\eal
where $\bh_{im}$ is the channel vector of UCA $m$ for user $i$, $\otimes$ denotes Kronecker product, $\bh_{ih}=\bh_{i1}$ is the 1st UCA streeting vector, $\bh_{iv}$ is the steering vector of vertical $N_c$-element ULA whose entries $h_{ivm}$ are the respective phase shift terms (which account for phase shifts of UCA $m$ with respect to 1st UCA):
\bal
h_{ivm} = \exp\{j 2\pi (m-1) d_v \cos(\theta_i)\},\ m=1...N_c
\eal
The interference leakage factors from $i$-th to main user are denoted as follows:
\bal\notag
\alpha_i = \frac{1}{N_cN}\bh_1^+\bh_i,\ \alpha_{ih} = \frac{1}{N}\bh_{1h}^+\bh_{ih},\ \alpha_{iv} = \frac{1}{N_c}\bh_{1v}^+\bh_{iv}
\eal
i.e. $\alpha_{h}$ and $\alpha_{v}$ represent the respective factors for a single UCA and ULA (located in horizontal and vertical planes, respectively).
The following Proposition is instrumental in establishing FP for the UCLA.

\begin{prop}
Let $\bh_i$ have Kronecker structure as in \eqref{eq.hi.Kron}. Then, $\alpha_i$ can be expressed as follows:
\bal
\label{xiMN}
\alpha_i=\alpha_{ih} \alpha_{iv}
\eal
and, furthermore,
\bal
\label{eq.xiMN.ineq}
|\alpha_i|\le \min\{|\alpha_{ih}|, |\alpha_{iv}|\}
\eal
\begin{proof}
Observe the following:
\bal\notag
\alpha_i &= \bh_1^+\bh_i/(N_c N)\\ \notag
&= (\bh_{1h} \otimes \bh_{1v})^+ (\bh_{ih} \otimes \bh_{iv})/(N_c N)\\ \notag
&=(\bh_{1h}^+ \otimes \bh_{1v}^+) (\bh_{ih} \otimes \bh_{iv})/(N_c N)\\ 
&=(N_c^{-1} \bh_{1v}^+ \bh_{iv})(N^{-1} \bh_{1h}^+ \bh_{ih})=\alpha_{ih} \alpha_{iv}
\eal
where 3rd and 4th equalities are due to the properties of Kronecker products \cite{Magnus-99}. The inequality in \eqref{eq.xiMN.ineq} follows from $|\alpha_{h,v}| \le 1$.
\end{proof}
\end{prop}

We note that this result holds for any array satisfying the Kronecker product property in \eqref{eq.hi.Kron}, not only cylindrical (e.g. it also holds for a planar array) and it can be extended to Kronecker products of any number of vectors, not just 2.

We are now ready to establish the FP property for the UCLA.

\begin{prop}
\label{prop.UCLA}
Favorable propagation holds for the UCLA in the full 3-D case under LOS propagation for (any) fixed element spacings $d, d_v >0$,  any fixed number of users $M < \infty$ with  distinct and fixed AoA's in any of the two cases:
\bi
\item Case 1. $N\to \infty$ and,  for all $i\ge 2$, neither of the following holds :

$\quad$ (a) $\theta_{1,i} = 0$ or $\pi$;

$\quad$ (b) $\phi_1 = \phi_i$ and $\theta_1=\theta_i$ or $\theta_1=\pi-\theta_i$,

\item Case 2. $N_c\to \infty$, $\theta_1\neq \theta_i$ for any $i\ge2$, $d_v < \gl/2$.
\ei
\end{prop}

\begin{proof}
Consider case 1. Using \eqref{eq.hi.Kron} and \eqref{eq.xiMN.ineq}, one obtains:
\bal
\lim_{N\to \infty} |\alpha_i| \le \lim_{N\to \infty} |\alpha_{ih}| = 0
\eal
from which the desired result follows, where the equality is due to Proposition \ref{prop.3D}. Case 2 follows in a similar way (from the FP property for the ULA).
\end{proof}

Based on the above result, we conclude that, in order to achieve favorable propagation, one should
\begin{itemize}
\item expand UCLA horizontally ($N\to \infty$) if $\phi_1\neq\phi_i$ or if $\theta_1\neq\theta_i$ and $\theta_1+\theta_i\neq\pi$, or
\item expand UCLA vertically ($N_c\to \infty$) if $\theta_1 \neq \theta_i$, $d_v < \gl/2$.
\end{itemize}
Clearly, simultaneous expansion (both vertically and horizontally) is also possible to achieve the FP.

\section{Examples}

To illustrate the asymptotic convergence to favorable propagation, 
Fig. 4 and 5 show the behavior of $|\alpha_N|$ and $|\alpha_{2N}|$ for the UCA (the planar 2-D case) with $d=0.5$,  $M=10$ and 100,  as the number $N$ of antennas is increasing; $|\alpha_{2N}|$ also corresponds to 2-user case with $\phi_1=0,\ \phi_2=2\pi/M$. Note that the envelope of $\alpha$ is monotonically decreasing with $N$, i.e. the FP property holds, and that larger $M$ calls for larger $N$ to attain low $|\alpha|$, in accordance with (10) and (11).

These simulation results are representative (in the sense that similar behaviour also holds for other values of $M$ and $d$); they are consistent with the presented analytical results (when FP holds, the respective limit is zero so that arbitrary low $|\alpha_N|$ can be achieved if $N$ is large enough). Additionally, they are also consistent with empirical measurements in e.g. \cite{Hoydis-12}\cite{Gao-15}.

\begin{figure}[t]
\centerline{\includegraphics[width=4in]{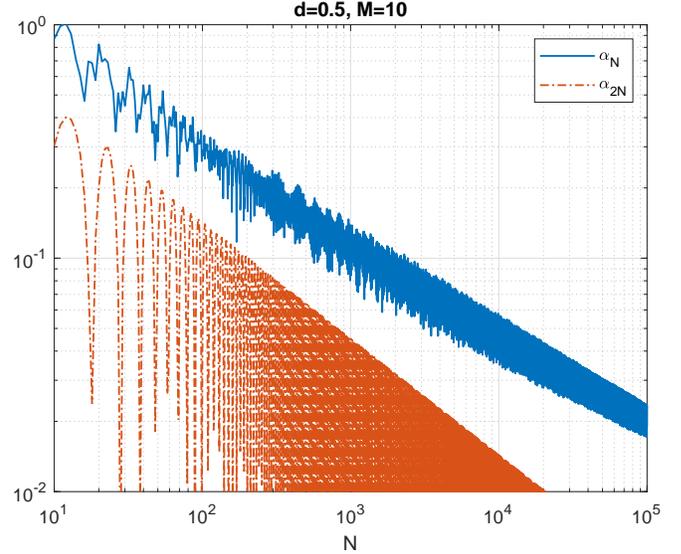}}
\caption{IUI factors $|\alpha_{2N}|$ and $|\alpha_N|$ vs. $N$ for $M=10$, $d=0.5$; uniform user AoAs $\phi_i = 2\pi(i-1)/M,\ i=1...M$.}
\label{fig.d05M10}
\end{figure}

\begin{figure}[t]
\centerline{\includegraphics[width=4in]{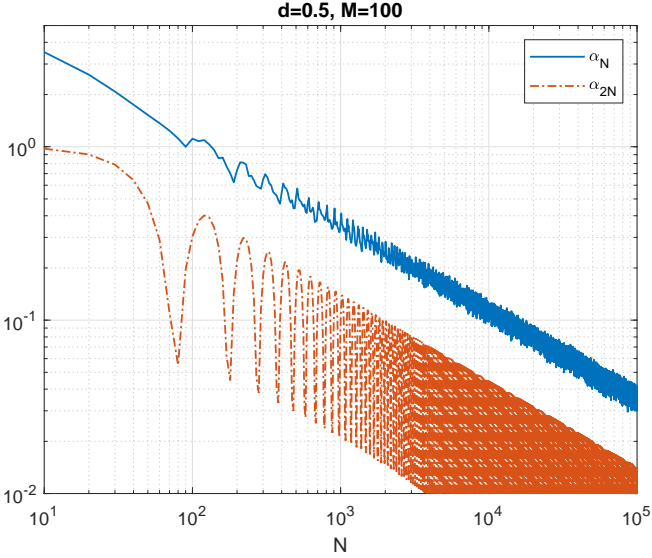}}
\caption{The same setting as in Fig. \ref{fig.d05M10} but with $M=100$ users. Note that larger $N$ (approx. 10 times) is needed to achieve the same low $\alpha$ compared to Fig. \ref{fig.d05M10}, in accordance with (10) and (11).}
\end{figure}

\section{Conclusion}
Favorable propagation is rigorously shown to hold in LOS environment for uniform circular and cylindrical arrays with (any) fixed antenna spacing if a given number $M$ of users is finite and their angles of arrival are fixed and distinct. Under these conditions, inter-user interference can be made as low as desired if the number $N$ of antennas is large enough. Based on the asymptotic analysis, a condition on $N$ to approach the FP closely is given. If either $M$ grows unbounded with $N$ or if the AoAs are allowed to approach each other, the FP does not hold in general but may hold in some special cases.

\section{Appendix: Proof of Theorem 1}

One can assume, without loss of generality (due to rotational symmetry), that $\phi_1=0$ and $\phi_i \neq 0,\ i\ge 2$, so that $\alpha_{iN}$ can be expressed as follows:
\bal
\label{eq.FP-2D-p-1}
\alpha_{iN} &= \frac{\bh_1^+\bh_i}{N }=\frac{1}{N}\sum_{n=0}^{N-1} e^{j2\pi R_N[\cos(\varphi_n - \phi_i) - \cos(\varphi_n)]}\\ \notag
&=\frac{1}{N}\sum_{n=0}^{N-1} e^{jz_{Ni}\sin(\varphi_n - \phi_i/2)}, \ z_{Ni}=4\pi R_N\sin(\phi_i/2)
\eal
where $\varphi_n = 2\pi n/N$. Note that this sum does not converge to an integral since $|z_{Ni}| \to \infty$ as $N \to \infty$ and hence the method of \cite{Chen-13} is not applicable.
To overcome this difficulty, we use a different approach, which is based on the generating function of the Bessel coefficients \cite[sec. 2.22]{Watson}:
\bal
e^{j x \sin(\phi_i)}=\sum_{m=-\infty}^{+\infty}e^{j m\phi_i}J_m(x)
\eal
where $J_m$ is Bessel function of 1st kind and order $m$, so that $\alpha_{iN}$ in \eqref{eq.FP-2D-p-1} can be expressed as follows:
\bal\notag
\alpha_{iN} &= \frac{1}{N} \sum_{m=-\infty}^{+\infty} e^{-j \phi_i m/2} J_{m}(z_{Ni}) \sum_{n=0}^{N-1} e^{j 2\pi m n /N}\\
\label{eq.FP-2D-p-32}
&= \sum_{k=-\infty}^{+\infty} e^{-j k N \phi_i/2} J_{k N}(z_{Ni})
\eal
where \eqref{eq.FP-2D-p-32} follows from
\bal
\sum_{n=0}^{N-1}e^{j 2\pi n m /N}=\left\{\begin{array}{ll}
N &\text{if}\ \ m=k N, \\
0 &\text{if}\ \ m \neq k N
\end{array}\right.
\eal
where $k=0,\pm 1,\pm 2, \cdots $. To facilitate limit evaluation, we use the following upper bound based on \eqref{eq.FP-2D-p-32} and the symmetry property $J_{-k}=(-1)^k J_k$:
\bal\notag
\label{eq.FP-2D-p-4}
|\alpha_{iN}| &\le 2 \sum_{k=0}^{k_1-1} |J_{k N}(z_{Ni})| + 2 \sum_{k=k_1}^{\infty} |J_{k N}(z_{Ni}) |\\
&= A_{1Ni} + A_{2Ni}
\eal
where we set $k_1=\lceil e\pi d\rceil$ (the usefulness of this choice will be clear later on), $\lceil x \rceil$ is the smallest integer greater or equal to $x$; $A_{1Ni},\ A_{2Ni}$ denote 1st and 2nd summation terms. Next, we obtain sufficiently-tight upper bounds on $A_{1Ni},\ A_{2Ni}$ and demonstrate that they converge to 0 as $N\to \infty$. To upper bound $A_{1Ni}$, use
\bal
z_{Ni} = 4\pi R_N\sin(\phi_i/2) = \frac{2Nd\sin(\phi_i/2)} {\operatorname{sinc}(1/N)}
\eal
where $\operatorname{sinc}(x)= \sin(\pi x)/(\pi x)$, and observe that
\bal
2/\pi \le \operatorname{sinc}(1/N) \le 1
\eal
so that
\bal
\label{eq.zN.LB}
2 N d |\sin(\phi_i/2)| \le |z_{Ni}| \le N\pi d
\eal
Combining \eqref{eq.zN.LB} with the following upper bound  \cite{Landau}\cite{Krasikov-06}
\bal
|J_k (x)|\le|x|^{-1/3}
\eal
one obtains
\bal
|J_{k N}(z_{Ni})|\le |z_{Ni}|^{-1/3}\le |2N d \sin(\phi_i/2)|^{-1/3}
\eal
and an upper bound on $A_{1Ni}$ follows:
\bal\notag
\label{eq.A1N.UB}
A_{1Ni} &\le 2\sum_{k=0}^{k_1-1}|2N d \sin(\phi_i/2)|^{-1/3}\\
 &=2 k_1 |2N d \sin(\phi_i/2)|^{-1/3}
\eal

Next, we obtain an upper bound on $A_{2Ni}$. To this end, we will need the following technical Lemma, which presents a novel upper bound on Bessel functions.

\begin{lemma}
\label{lemma.UB}
If $|x|\le 1$, then
\bal
\label{eq.Jn.UB}
|J_{n}(n x)| \le |x\cdot e/2|^n
\eal
\end{lemma}
\begin{proof}
via the following inequalities:
\bal\notag
\label{eq.Jn.UB.p1}
|J_n(n x)| &\le \left|\frac{x\cdot \exp\left\{\sqrt{ 1-x^2}\right\}}{1+\sqrt{1-x^2}}\right|^n\\
 &\le |x|^{n} \cdot \max _{|t|\le 1}\left(\frac{e^{\sqrt{1-t^2}}}{1+\sqrt{1-t^2}}\right)^{n}\\ \notag
  &\le |e\cdot x/2|^n
\eal
where the first inequality in \eqref{eq.Jn.UB.p1} is
Kapteyn's inequality \cite[sec.~8.7]{Watson} and the last inequality is due to the following:
\bal
\max _{|t|\le1}\frac{e^{\sqrt{1-t^2}}}{1+\sqrt{1-t^2}} =\frac{e}{2}
\eal
\end{proof}

Using \eqref{eq.Jn.UB} and  \eqref{eq.zN.LB}, one obtains
\bal
|J_{k N}(z_{Ni})| \le \left(\frac{e \pi d}{2 k}\right)^{k N}
\eal
so that $A_{2Ni}$ can be upper bounded as follows:
\bal\notag
\label{eq.A2N.UB}
A_{2Ni} & \le 2\sum_{k=k_1}^{\infty}\left(\frac{e \pi d}{2 k}\right)^{k N} \\ \notag
&\le 2\sum_{k=k_{1}}^{\infty} q^{N k}, \quad q=\left(\frac{e \pi d}{2 k_{1}}\right)\\
&\le 2\sum_{k=k_{1}}^{\infty} (1/2)^{N k}= \frac{2 (1/2)^{N k_1}}{1-(1/2)^{N}}
\eal
Using \eqref{eq.A1N.UB} and \eqref{eq.A2N.UB}, one finally obtains
\bal
\label{eq.alpaN.UB}
\lim _{N \to \infty} &|\alpha_{iN}| \le \lim_{N \to \infty} (A_{1Ni}+A_{2Ni})\\ \notag
 &\le \lim_{N \to \infty} \frac{2 \lceil e\pi d\rceil}{|2N d \sin(\phi_i/2)|^{1/3}} +\lim _{N \to \infty}\frac{2 (1/2)^{N \lceil e\pi d\rceil}}{1-(1/2)^{N}}=0
\eal
and therefore $\lim _{N \to \infty} |\alpha_{iN}| =0$, which implies $\lim _{N \to \infty} |\alpha_{N}| =0$, as required.



\end{document}